%% file: main.tex
\documentclass[a4paper]{article}
\input{lnnote}

\newcommand{\E}{\mathbb{E}}
\newcommand{\ceil}[1]{\ensuremath{\lceil #1 \rceil}}
\newcommand{\floor}[1]{\ensuremath{\lfloor #1 \rfloor}}

\newcommand{\NLT}{\textnormal{\textbf{NLT}}}
\newcommand{\RW}{\textnormal{\textbf{RW}}}
\newcommand{\PE}{\textnormal{\textbf{PE}}}

\begin{document}

\title{Influence Maximization under The Non-progressive Linear Threshold Model
}

\author{
		T-H. Hubert Chan  \\
        Department of Computer Science \\
		the University of Hong Kong \\ 
		Pokfulam Road, Hong Kong\\
        {\small hubert@cs.hku.hk}           %  \\
		\and
        Li Ning$^*$ \\ 
        Shenzhen Institutes of Advanced Technology\\
		Chinese Academy of Science \\
		Xueyuan Road 1068, Shenzhen, China \\
		{\small li.ning@siat.ac.cn}
}

\maketitle

\begin{abstract}
In the problem of influence maximization in information networks, the objective
is to choose a set of initially active nodes subject to some budget constraints
such that the expected number of active nodes over time is maximized. 

The linear threshold model has been introduced to study the opinion cascading behavior,
for instance, the spread of products and innovations.
In the existing studies, the study of the linear threshold model mainly focus on the progressive case, 
in which once a user became active, it is forced to be active forever. 
In this paper, we consider the non-progressive case in which active nodes might become
inactive in subsequent time steps. This setting makes it possible to model the users' dynamic behavior, 
and consequently fit better to the situation of continuous consumption in the daily life. 
Previous works on the non-progressive case assumed that the 
thresholds indicating the susceptibilities of the individuals change
randomly and independently at every time step. We argue that an individual's 
susceptibility should be consistent, and hence it is more realistic
to consider the case in which an individual's susceptibility is chosen
initially at random, but then remains the same throughout the process.
This setting causes more completeness for the analysis, as for any two nodes
that share a common ancestor, their status are no longer independent.

In this paper, we we extends the classic 
\emph{linear threshold model} \cite{Kempe2003} to capture the non-progressive behavior.
The information maximization problem under our model is proved to be NP-Hard, 
even for the case when the underlying network has no directed cycles. 
The first result of this paper is negative. 
In general, the objective function of the extended linear threshold model is no longer submodular,
and hence the \emph{hill climbing} approach that is commonly used in 
the existing studies is not applicable.
Next, as the main result of this paper, we prove that if the underlying information network
is directed acyclic, the objective function is submodular (and monotone). 
Therefore, in directed acyclic networks with a specified budget we can achieve $\frac{1}{2}$-approximation
on maximizing the number of active nodes over a certain period of time by a deterministic algorithm, and
achieve the $(1-\frac{1}{e})$-approximation by a randomized algorithm.

\end{abstract}

\input{intro.tex}

\input{pre.tex}
\input{hardness.tex}

\input{negative.tex}

\input{acyclic.tex}

\bibliographystyle{abbrv}
\bibliography{ref}

\end{document}

%% file: lnnote.tex
\usepackage{amsmath}\allowdisplaybreaks[4]
\usepackage{amsthm}%for proof
\usepackage[ruled,vlined]{algorithm2e}
\usepackage{graphicx}
\usepackage{amsfonts}
\usepackage{amssymb}
\usepackage[sl]{caption}
\usepackage[top=3cm,bottom=3cm,left=2cm,right=2cm]{geometry}
\newcommand\figcaption{\def\@captype{figure}\caption}%
\newcommand\tabcaption{\def\@captype{table}\caption}%
\makeatother

\usepackage{fancyhdr}
\usepackage[scriptsize]{subfigure}

\usepackage{color}
\usepackage{listings}
\newtheorem{theorem}{\hskip\parindent Theorem}
\newtheorem{lemma}{\hskip\parindent Lemma}
\newtheorem{remark}{\hskip\parindent Remark}
\newtheorem{definition}{\hskip\parindent Definition}
\newtheorem{corollary}{\hskip\parindent Corollary}

\newtheorem{model}{\hskip\parindent Model}

\newcounter{note}[section]

\newcommand{\li}[1]{\refstepcounter{note}$\ll${\sf Li's
Comment~\thenote:} {\sf \textcolor{blue}{#1}}$\gg$}

\newcommand{\ignore}[1]{}

%% file: intro.tex
\section{Introduction}

We consider the problem of an advertiser promoting a product in a social network.  
The idea of \emph{viral marketing}~\cite{Domingos2002,Jurvetson2000} is that
with a limited budget, the advertiser can persuade only a subset of individuals to use the new product, 
perhaps by giving out a limited number of free samples.
Then the popularity of the product is spread by word-of-mouth, 
i.e.\ through the existing connections between users in the underlying social network.

Information networks have been used to model such cascading behavior~
\cite{Wasserman1994,Domingos2002,Domingos2001,Goldenberg2001b,Goldenberg2001,Kempe2003,Kempe2005}. 
An information network is a directed edge-weighted graph, in which a node represents a user, 
whose behavior is influenced by its (outgoing) neighbors,
and the weight of an edge reflects how influential the corresponding neighbor is.  A node adopting
the new behavior is \emph{active} and is otherwise \emph{inactive}.
The \emph{threshold model}~\cite{Granovetter1978,Kempe2003}
is one way to
model the spread of the new behavior.  The resistance of a node $v$ to adopt the new behavior is represented
by a random threshold $\theta_v$ (higher value means higher resistance), where the randomness is used to 
model the different susceptibility of different users.
The new behavior is spread in the information network in discrete time steps.  An inactive node
changes its state to active if the weighted influence from the active neighbors in the previous time 
step reaches its threshold. We consider the \emph{non-progressive} case where an active node could 
revert back to the inactive state if the influence from its neighbors drops below its threshold.

\noindent \paragraph{\textbf{\emph{Our Contribution.}}}
We consider the non-progressive linear threshold model in this paper, which is 
the natural extension of the well known linear threshold model \cite{Kempe2003}.
In Section~\ref{sec:pre}, the formal definition of the non-progressive
linear threshold model is introduced, as well as the influence maximization problem.
In most existing works, for a set of initially active nodes, the influence is measured by 
the maximum number of active nodes. Since the existing works consider the progressive case, 
hence the number of active nodes increases step by step and achieves the maximum after at most $n$ time steps, 
where $n$ is the number of the nodes. However, in the non-progressive case, it is possible 
that the active status never become stable. Hence, we introduce the average number of active nodes
over a time period to measure the influence. Similar to the progressive case, 
the influence maximization problem considering the non-progressive linear threshold model is also
NP-hard (Section~\ref{app:hardness}). 
In order to approximate the optimal within a constant factor, a commonly used approach 
is to prove the monotone and submodular property of the objective function. Then the constant approximation
ratio algorithms are promised by using the results of Fisher et al.\ \cite{Fisher1978} and Calinescu et al.\ 
\cite{DBLP:journals/siamcomp/CalinescuCPV11}. However, this approach
is not generally applicable for the non-progressive linear threshold model, since
as showed in Section~\ref{appendix:nonsub}, the average number of active nodes is possibly not submodular.
As the main result of this paper, we studied the case when the information network is acyclic. 
As consistent with the intuition, the expected influence under the acyclic networks is submodular
and hence Fisher's technique (and Calinescu's technique) is applicable to achieve the constant approximation.
It should be noted that although the acyclic case looks much simpler than the general ones (where directed cycles
may exist), the solution to maximize the expected influence is not that easy. As it is proved
in Section~\ref{app:hardness}, the problem of influence maximization is still NP-Hard even for the 
case under acyclic information networks. Futhermore, to prove the submodularity of the expected influence, 
we still need some tricky technique (in this paper, we ) to handle complicated association between the status of the nodes. 
To see the association, consider time $t > 0$, the status of nodes at time $t$ are associated if
they share some common ancestors and the threshold of such an ancestor affects the status of its descendants. 
As this kind of association exist, it requires more carefully consideration of the nodes status and 
the analysis consequently become more complicate. In Section~\ref{sec:acyclic},
we introduce an equivalent process (called Path Effect) for the non-progressive linear threshold model, and
then the deep connection between this process and the random walk is proved via a coupling technique,  
which consequently leads to our final conclusion of the submodularity of the expected influence (under non-progressive linear threshold model). 

\noindent \paragraph{\textbf{\emph{Related Works.}}}
The cascading behavior in information networks was first studied
in the computer science community by Kempe, Kleinberg and Tardos~\cite{Kempe2003}.
They considered the Independent Cascade Model and the Linear Threshold Model,
the latter of which we generalize in this paper.  Their main focus was
the progressive case, and only reduced the non-progressive case
to the progressive one by assigning a new independent random threshold
to each node at every time step  such that the resulting
objective function is still submodular. 

%About the non-progressive case of Linear Threshold Model, one concrete model
%has been studied in \cite{ref:Kempe2003}. In this model, each individual's threshold
%is changed as time goes by. Because intuitively the threshold represents the
%hardness to convince an individual, then it is not very comfortable if we say
%an individual is easy to be convinced on one day and hard to be convinced on
%another day. For this reason,
%we consider the kind of non-progressive case that any individual's threshold is
%unchanged once it is chosen.

Kempe et al.~\cite{Kempe2003,Kempe2005} have also shown that the influence maximization problem
in such models is NP-hard.  Researchers often first show that
the objective functions in question are submodular and then apply 
submodular function maximization methods 
to obtain constant approximation ratio.
An example of such methods is the \emph{Standard Greedy Algorithm}, which is analyzed by
by Nemhauser and Fisher et al.~\cite{Nemhauser1978,Fisher1978}.
Loosely speaking, the Standard Greedy Algorithm (also known as the
\emph{Hill Climbing Algorithm}) starts with an empty solution, and
in each iteration
while there is still enough budget, we expand the current solution
by including an additional node that causes the greatest increase in the objective function.
Although the costs for transient and permanent nodes are different in our model,
the budget constraint can still be described by a matroid.
Under the matroid constraint, Fisher et al.~\cite{Fisher1978}
showed that the Standard Greedy Algorithm achieves $\frac{1}{2}$-approximation,
and  Calinescu et al.~\cite{DBLP:journals/siamcomp/CalinescuCPV11} 
introduced a randomized algorithm that
achieves $(1-\frac{1}{e})$-approximation in expectation.

In the above submodular function maximization algorithms,
the objective function needs to be accessed in each iteration.
However, to calculate the exact value of the objective function is in general hard~\cite{Chen2010}.
One way to resolve this is to estimate the value of the objective function by sampling.
Some works have used other ways to overcome this issue.
To improve the efficiency of the Standard Greedy Algorithm, Leskovec et al. \cite{Leskovec2007}
showed a \emph{Cost-Effective Lazy Forward} scheme, which makes use of the submodularity of the objective function and
avoids the evaluation of influence on those nodes
for which the incremental influence
in the previous iteration is less than that of some already evaluated node in the current iteration.
This scheme has been shown more efficient than the Standard Greedy Algorithm by experiments.
%
%Some works have been done to overcome this issue by implementing the greedy algorithm cleverly \cite{Leskovec2007} or designing
%heuristic method to achieve similar performance \cite{Chen2009}. 
Chen et al. \cite{Chen2009} also designed an improved scheme to speed up the Standard Greedy Algorithm
by using some efficiently computable heuristics that have similar performance.

%.
%Moreover, in \cite{Chen2009} they have mentioned the idea to discovery heuristics with comparable
%performance with the greedy hill climbing, and shown the \emph{Degree Discount Heuristic}, in which
%the already selected targets would be discounted in the calculation of degrees on the other nodes.
%%The idea to ``discount'' something is also the base of our \emph{Greedy Centrality Heuristic} method. 

Chen et al. \cite{Chen2011} considered how \emph{positive} and \emph{negative} opinions 
spread in the same network, which can be interpreted as the influence process involving two agents.
The influence maximization problem considering multiple competing agents in an information network
has also been studied in \cite{Hotelling1929,Dubey2006,Bharathi2007,Carnes2007}. 
We follow a similar setting in which a new comer can observe the strategies of existing agents,
and stategizes accordingly to maximize his influence in the network.

Mossel and Roch~\cite{Mossel2007} have shown that under more general submodular threshold functions
(as opposed to linear threshold functions), the objective function is still submodular
and hence the same maximization framework can still be applied.

Similar to our approach, the relationship between influence spreading and random walks has been
investigated by Asavathiratham et al.~\cite{Asavathiratham2001}, and Even-Dar and Shapira~\cite{EvenDar2007} in other information network models.

%% file: pre.tex
\section{Preliminaries}\label{sec:pre}

\begin{definition}[Information Network]
An information network is a directed weighted graph $G=(V,E,b)$ with node set $V$ and edge set $E$,
where each edge $(v,u)\in E$ has some positive weight $0 < b_{vu} < 1$ 
which intuitively represents the influencing power of $u$ on $v$.
Denote the set of outgoing neighbors of a node $v$ by $\Gamma(v) := \{u\in V \mid  (v,u) \in E\}$.
In addition, for each $v \in V$, the total weight of its outgoing edges is at most $1$, 
i.e.\ $\sum_{u\in\Gamma(v)} b_{vu} \leq 1$.
\end{definition}

Without loss of generality, we assume that in the considered information network $G$,
$\sum_{u\in \Gamma(v)}b_{vu}$ is exactly $1$ for every node $v\in V$.
To achieve this requirement, for any given $G$, we add a void node $d$
and for each node $v \neq d$, we include $(v,d)$ into the set of edges,
and set $b_{vd}:=1-\sum_{u\in \Gamma(v) \setminus \{d\}}b_{vu}$.  
We can add a self loop with weight $1$ at the void node $d$.
Furthermore, $d$ is never allowed to be active initially, and hence will never be active. 
Unless explicitly specified, when we use the term \emph{node} in general, we mean a node other than the void node.

Next, we formally describe the extension of the classic linear threshold model for the adaption of the non-progressive behavior. 
A new feature of our model is that an initially active node can be either \emph{transient} or \emph{permanent}.

\begin{model}[Non-progressive Linear Threshold Model (\NLT)]\label{model:avg}
Consider an information network $G=(V,E,b)$. Each node in $V$ is associated
with a threshold $\theta_v$, which is chosen from $(0,1)$ independently and uniformly at random.
At time $t \geq 0$, every node $v$ is either \emph{active} $\mathcal{A}$ or \emph{inactive} $\mathcal{N}$.
Denote the set of active nodes at time $t$ by $A_t$. In the influence process, 
given a transient initial set $A\subseteq V$, a permanent initial set $\widehat{A}\subseteq V$,
and a configuration of thresholds $\mathbb{\theta} = \{\theta_v\}_{v\in V}$, the nodes update
their status according to the following rules. 
\begin{itemize}
\item[1.]
At time $t=0$, $A_0:= A \cup \widehat{A}$. 
\item[2.]
At time $t>0$, for each node $v\in V\setminus \widehat{A}$, compute the activation function 
$f_v(A_{t-1}) := \sum_{u\in A_{t-1}\cap \Gamma(v)} b_{vu}$. Then let 
$A_t := \{v\in V\setminus \widehat{A} | f_v(A_{t-1})\geq \theta_v\}\cup \widehat{A}$.
%\item[3.]
%If $x_t=x_{t-1}$ happens at any $t>0$, the procedure \emph{stabilizes}. \hubert{Why do we need this?}
%\li{Not necessarily.}
\end{itemize}
\end{model}

Without loss of generality, we can assume $A \cap \widehat{A} = \emptyset$, 
otherwise we can use $A \setminus \widehat{A}$ as the transient initial set instead. 
Given a transient initial set $A$ and a permanent initial set $\widehat{A}$, 
we measure the influence of the agent by the average number of active nodes over $T$ time steps, 
where $T$ is some pre-specified time scope in which the process evolves.
Observe that once the initial sets and the configuration of the thresholds are given, the
active sets $A_t$'s are totally determined.

\begin{definition}[Influence Function and Expected Influence]
Given an information network $G$, a transient initial set $A$,
a permanent initial set $\widehat{A}$, and a configuration $\theta$ of thresholds,
the average influence over time period $[1,T]$ is defined as
$\sigma^{[1,T]}_\theta(A, \widehat{A}) := \frac{1}{T} {\sum^{T}_{t=1}\mid A_t\mid}$.
For simplicity, we ignore the superscript $[1,T]$ in $\sigma$ when the target period is clear from
the context. We define the expected
influence $\overline{\sigma}$ as the expectation of $\sigma_\theta(A, \widehat{A})$
over the random choice of $\theta$, i.e.,
$ \overline{\sigma}(A, \widehat{A}):=\mathbb{E}_\theta[\sigma_\theta(A, \widehat{A})]$.
\end{definition}

\ignore{
\begin{remark} We will later restrict our attention to acyclic network for the 
\NLTd model.  In this case, the effect of the initial transient nodes will eventually vanish, as $T$ gets larger than $n$.  Hence, in order to measure the transient effect of the transient nodes in an acyclic network, we assume
the expected influence is considered over a small number $T$ of time steps.
\end{remark}
}
\begin{definition}[Influence Maximization Problem]
In an information network $G$, suppose
the advertising cost of a transient initial node is $c$ and that of a permanent initial node is $\widehat{c}$, where the costs are uniform
over the nodes. Given a budget $K$, the goal is to
find a transient initial set $A$ and a permanent initial set $\widehat{A}$ 
with total cost $c \cdot |A| + \widehat{c} \cdot |\widehat{A}|$ at most $K$
such that $\overline{\sigma}(A, \widehat{A})$ is maximized.
\end{definition}
The most technical part of the paper is to show that
the objective function $\overline{\sigma}$ is submodular
so that the maximization techniques of Fisher et al.~\cite{Fisher1978}
can be applied.
%\begin{definition}[Monotone]\label{def:monotone}
%A function $f:2^V \rightarrow \mathbb{R}$ is monotone if
%for any $A\subseteq A'\subseteq V$,  $f(A)\leq f(A')$ holds.
%\end{definition}
\begin{definition}[Submodular, Monotone]\label{def:submodular}
A function $f:2^V\rightarrow\mathbb{R}$ is submodular if
for any $A\subseteq B\subseteq V$ and $w \in V \setminus B$,
$f(B\cup\{w\})-f(B)\leq f(A\cup\{w\})-f(A)$ holds.
A function $f$ is monotone if for any $A\subseteq B$,
$f(A) \leq f(B)$.
A function
$g : 2^V \times 2^V \rightarrow \mathbb{R}$ is submodular (monotone), if
keeping one argument constant, the function is submodular (monotone) as a function on
the other argument.
\end{definition}
%\begin{definition}[Additive]\label{def:add}
%A function $f:2^V\rightarrow\mathbb{R}$ is additive if
%for any
%$A\subseteq V$, $f(A)=\sum_{u\in A} f(\{u\})$ holds.
%\end{definition}
%We first consider the simple case in which all the initially active nodes are transient, i.e.
%$\widehat{A} = \emptyset$.  In this case, the expected influence is a function of the set $A$ only.
%
%
%
%At the end of this section we introduce indicator variables, which are useful for discussing the
%expected influence.
In order to facilitate the analysis of the influence process,
we define indicator variables to consider the behavior
of individual nodes at every time step.
\begin{definition}[Indicator Variable]
In an information network $G$, given a transient initial set $A$
and a permanent initial set $\widehat{A}$, a node $v$ and a time $t$, let
$X^t_v(A,\widehat{A})$ be the indicator random variable
that takes value 1 if node $v$ is active at time $t$, and 0 otherwise.
%More precisely, given $A$ and $\theta$, $X^t_v(A,\theta):=1$ if $v$ is active at time $t$, and $X^t_v(A,\theta):=0$ otherwise.
%If we consider the indicator variable while keeping $\theta$ in randomness,
%then $X^t_v(A,\theta)$ is a random variable for $\theta$.
When $\widehat{A} = \emptyset$, we sometimes write
$X^t_v(A) := X^t_v(A,\emptyset)$.
\end{definition}
The indicator variable's usefulness is based on the following equality:
\[\overline{\sigma}(A, \widehat{A}) = \mathbb{E}\left[\frac{1}{T} \sum^{T}_{t=1} |A_t|\right] =
\frac{1}{T} {\sum^{T}_{t=1}\sum_{v\in V}\mathbb{E}\left[X^t_v(A, \widehat{A})\right]}.\]
Hence, if the function $(A, \widehat{A}) \mapsto
\mathbb{E}[X^t_v(A, \widehat{A})]$ is submodular and monotone, then so is $\overline{\sigma}$.

\ignore{
At the end of this section, we introduce the
random walk process that we later prove is
highly related to our influence model.  

\begin{model}[Random Walk Process]\label{model:R}
Consider an information network $G$. For any given node $v\in V$, 
we define the Random Walk Process as follows.
\begin{itemize}
\item[1.]
At time $t=0$, the walk starts at node $v$.
\item[2.]
Suppose at some time $t$, the current node is $u$. 
A node $w\in \Gamma(u)$ is chosen with probability $b_{uw}$.
The walk moves to node $w$ at time $t+1$. \footnote{Observe that if
$w$ is the void node, then the walk remains at $w$.}
\end{itemize}
\end{model}
\begin{definition}[Reaching Event]\label{def:R}
Let $G$ be an information network.
For any node $v$, subset $C \subseteq V$ and $t \geq 0$, 
we use $R^t_v(C)$ to denote the event that a Random Walk Process on $G$ starting from $v$ 
would be at a node in $C$ at precisely time $t$.
\end{definition}
\begin{definition}[Passing-Through Event] \label{defn:passing}
Let $G$ be an information network. For any node $v$, subset $C \subseteq V$ and
$t \geq 0$, we use $S^t_v(C)$ to denote the event that a Random Walk Process on $G$
starting from $v$ would reach a node in $C$ at time $t$ \textbf{or before}.
\end{definition}
}
%%%%%%%%%%%%%%%%%%%%%%%%%%%%%%%%%%%%%%%%%%%%%%%%%%%%%%%%%%%%%%%%%%%%%%%%%%%%

%% file: hardness.tex
\section{Hardness of Maximization Problem}
\label{app:hardness}

We outline an NP-hardness proof for the maximization problem
in our setting via a reduction from vertex cover similar to that in (\cite{Kempe2003,Kempe2005}).
We show that the problem is still NP-hard, even for the special case when the network is acyclic, 
and each transient node and each permanent has the same cost, 
which means only permanent nodes will be used.

\begin{theorem} \label{th:hardness}
The influence maximization problem under the non-progressive linear 
threshold model is NP-hard even when the network 
is a directed acyclic graph and all the initially active nodes are permanent.
\end{theorem}

\begin{proof}
Given an undirected graph with $n$ vertices, we pick an arbitrary
linear ordering of the nodes and direct each edge accordingly to
form a directed acyclic graph. We add a dummy node and for nodes
with no outgoing edges, we add an edge from it to the dummy node.
Hence, the network has $n+1$ nodes.  For each node, the weights of
its outgoing edges are distributed uniformly.  The number $T$ of
time steps under consideration is $1$.

We claim that there is a vertex cover of size $k$ for the constructed network 
\emph{iff} there is a permanent initial set $\widehat{A}$ of size $k+1$ such
that $\overline{\sigma}(\emptyset, \widehat{A}) = n+1$.

Suppose there is a vertex cover $S$ of size $k$, 
then adding the dummy node to $S$ to form $\widehat{A}$ as the permanent
initial set, all nodes will be active in the next time step with probability $1$, 
and so $\overline{\sigma}(\emptyset, \widehat{A}) = n+1$.  
On the other hand, if the permanent initial set
$\widehat{A}$ has size $k+1$ and $\overline{\sigma}(\emptyset, \widehat{A})=n+1$, 
then the dummy node must be in $\widehat{A}$, and suppose $S$ the set of non-dummy nodes in $\widehat{A}$.  
If $S$ does not form a vertex cover for the given graph,
then there exists some edge $(u,v)$, where both nodes $u$ and $v$ are inactive initially, 
and hence the probability that $u$ is active in the next time step is strictly smaller than $1$.
\qed \end{proof}

%% file: negative.tex
\section{Information Network with Directed Cycles}
\label{appendix:nonsub}

In this section, we show that for cases where the information network has directed cycles, 
the expected influence function is not submodular in general. Before describing the example
that counters the submodularity, we introduce a conclusion that assists the argument. 

\begin{theorem}\label{thm:ind}
Given an information network $G$, if $A\mapsto\mathbb{E}[X^t_v(A)]$ is submodular 
for every node $v\in V$ and time $t > 0 $, then the expected influence is submodular.
On the other hand, if there exists a node $v^*\in V$ such that 
\[\frac{1}{T}{\sum^T_{t=1}\mathbb{E}[X^t_{v^*}(A)]}\]
is not submodular, then by modifying $G$, we can construct an information network, for which the
expected influence over period $[1,T]$ is not submodular.
\end{theorem}
\begin{proof}
The first statement is correct by observing that,
\begin{eqnarray*}
\mathbb{E}\left[\frac{1}{T}\sum^{T}_{t=1} |A_t| \right]&=&
\frac{1}{T}\sum^{T}_{t=1}\sum_{v\in V}\mathbb{E}[X^t_v(A)].
\end{eqnarray*} 
Consider the second statement. We add a set $L$ of nodes outside $V$ to the graph,
and for any node $v\in L$, we include an edge $(v,v^*)$ with weight $1$.  
Observe that this magnifies the effect of $v$ in the network.  When $|L|$ is large enough, 
the new expected influence is not a submodular function.
%
%
%
%Then the submodularity of the function
%$A\mapsto\sum^{T}_{t=s}\mathbb{E}[X^t_{v^*}(A)]/T$ is expended.
%When $\mid L\mid$ is large enough, this
%way of expending submodularity implies an information network for which the expected influence
%is not submodular.
\qed \end{proof}

Next, we are ready to describe the counter examples for the submodularity,

\begin{theorem}[Non-submodularity of Expected Influence]
There exists an information network with directed cycles, for which the expected
influence under \NLT~model is not submodular.
\end{theorem}
\begin{proof}
Based on Theorem~\ref{thm:ind}, it is sufficient to show that there
exists an information network, which contains a node $v^*$, such that
the function
$A\mapsto \frac{1}{T}  \sum^{T}_{t=1}\mathbb{E}[X^t_{v^*}(A)]$
is not submodular.

Consider the information network in Fig~\ref{fig1}(a). Each edge $(v,u)$ is
marked with its weight $b_{vu}$.

\begin{figure}[!htb]
\centering
\subfigure[directed graph with cycles]{\includegraphics[width=0.4\textwidth]{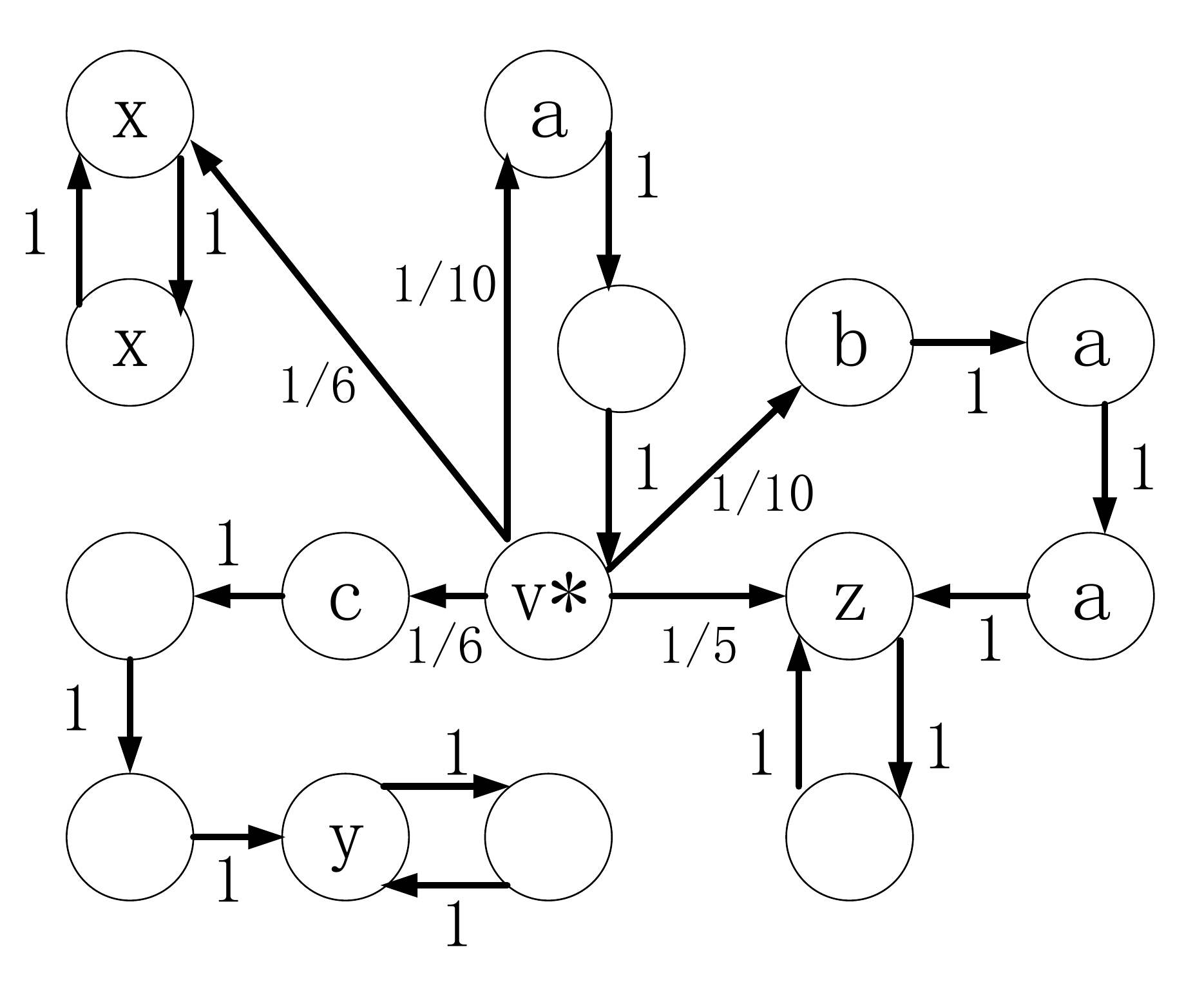}}
\subfigure[directed graph with self-loops]{\includegraphics[width=0.4\textwidth]{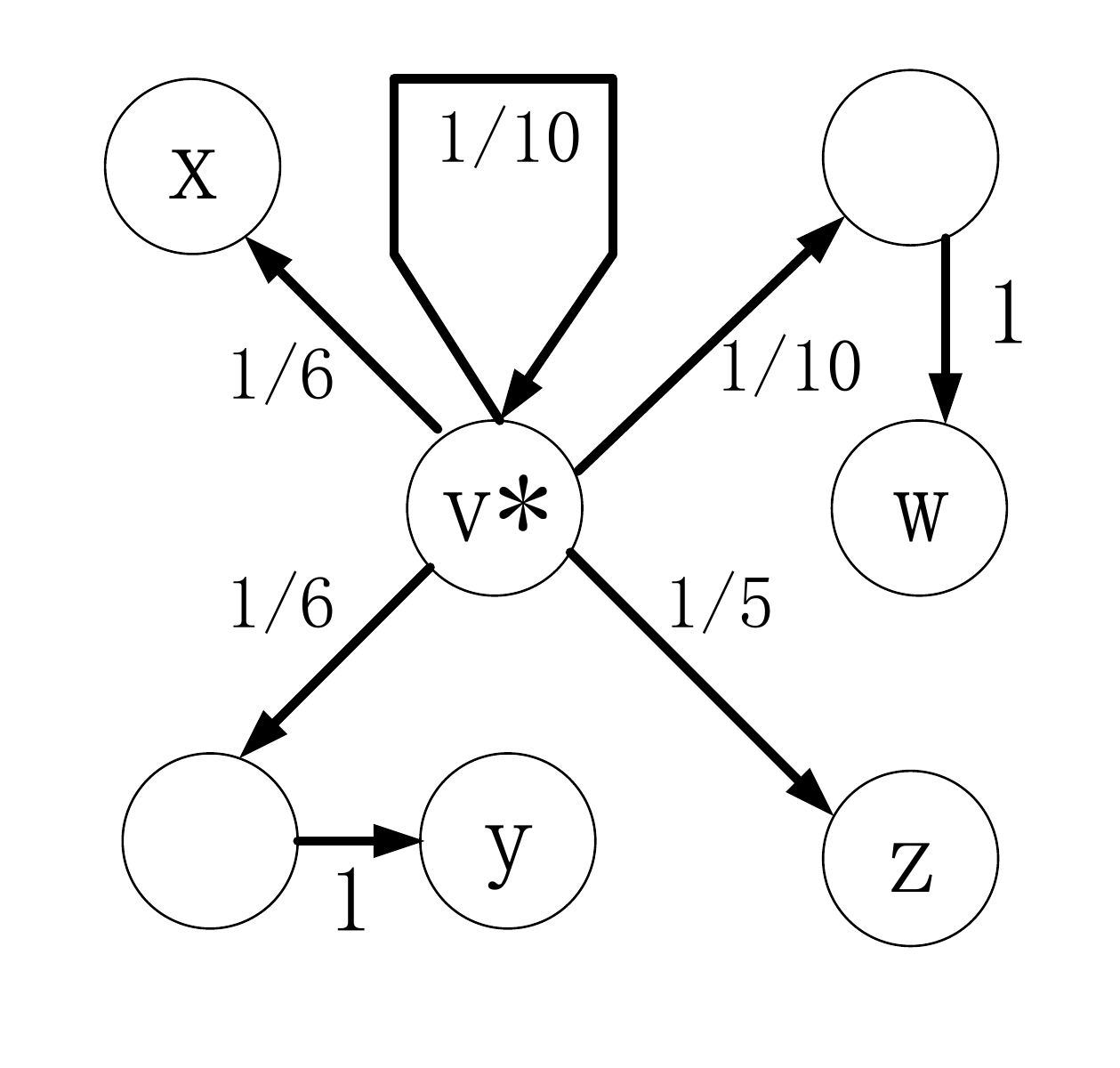}}
\caption{Counter examples for submodularity.}
\label{fig1}
\end{figure}

%    \caption{There exists an undirected graph with cycles, under which 
%	the expected influence of Model~\ref{model:I} is not submodular.
%    There exists an undirected graph with self-loops, under which 
%	the expected influence of Model~\ref{model:I} is not submodular.}

Let $S:=\{x\}$, $T:=\{x,y\}$ and focus on node $v^*$. Then
it is easy to check the following facts.
\begin{enumerate}
\item
For any even time $t\geq 4$,
$\mathbb{E}[X_{v^*}^t(S\cup\{z\})]-\mathbb{E}[X^t_{v^*}(S)]
<\mathbb{E}[X_{v^*}^t(T\cup\{z\})]-\mathbb{E}[X^t_{v^*}(T)]$.
\item
For any odd time $t\geq 5$,
$\mathbb{E}[X^t_{v^*}(S\cup\{z\})]-\mathbb{E}[X^t_{v^*}(S)]
=\mathbb{E}[X^t_{v^*}(T\cup\{z\})]-\mathbb{E}[X^t_{v^*}(T)]$.
\end{enumerate}
This implies the function $A\mapsto\sum^{T}_{t=1}\mathbb{E}[X^t_{v^*}(A)]/T$ is not
submodular when $T$ is large enough.
\qed \end{proof}

\begin{theorem}[Networks with Self-loops]
There exists an information network in which the only directed cycle
is a self-loop (on a non-void node), such that the 
according expected influence under \NLT~model %Model~\ref{model:avg} 
is not submodular.
\end{theorem}
\begin{proof}
Based on Theorem~\ref{thm:ind}, it is sufficient to show that there
exists an information network, which contains a node $v^*$, such that the function
$A\mapsto\sum^{T}_{t=1}\mathbb{E}[X^t_{v^*}(A)]/T$
is not submodular. Actually, we only need to find a network in which there
is a node $v^*$ and a particular time step $t$, such that $A\mapsto\mathbb{E}[X^t_{v^*}(A)]$
is not submodular.
Consider the information network in Fig~\ref{fig1}(b). Each edge $(v,u)$ is
marked with the weight $b_{vu}$. Let $S:=\{x\}$, $T:=\{x,w,z\}$,
and focus on node $v^*$. Then it is easy to check that for time $t=2$,
\begin{eqnarray*}
\mathbb{E}[X_{v^*}^t(S\cup\{y\})]-\mathbb{E}[X^t_{v^*}(S)]
<\mathbb{E}[X_{v^*}^t(T\cup\{y\})]-\mathbb{E}[X^t_{v^*}(T)].
\end{eqnarray*}

Thus, $A\mapsto\mathbb{E}[X^2_{v^*}(A)]$ is not submodular.
\qed \end{proof}

%% file: acyclic.tex
\section{Acyclic Information Networks}\label{sec:acyclic}

In this section, we consider the information networks without directed cycles. 
As we proved in Section~\ref{app:hardness}, 
the influence maximization problem under \NLT~model %Model~\ref{model:avg} 
is still NP-Hard even when the underlying network
has no directed cycles.

Note that with the assumption of acyclic information networks, for any node $v\in V$ other than the void node, 
the set of its outgoing neighbors $\Gamma(v)$ has no directed path back to $v$.
\footnote{Hence, the self-loop at the void node does not really
interfere with the acyclic assumption.}
Intuitively, during the influence procedure, 
$v$'s choice of threshold $\theta_v$ can never affect the states of nodes in $\Gamma(v)$.
To describe this fact formally, we introduce a random object on $\Omega$, 
where $\Omega$ is the sample space, which is essentially the set of all possible configurations of the thresholds.

\begin{definition}[States of Nodes over Time]\label{def:rand_ob}
%Suppose $\mathbb{T}=\{1,2,\ldots, T\}$ is the range of time steps in question.
Suppose $\Omega$ is the sample space, and $W$ is a subset of nodes.
%and the set $\mathbb{S} := \{\mathcal{A}, \mathcal{N}\}$ of states,
Define a random object
$\Pi_W:\Omega\rightarrow \{\mathcal{A}, \mathcal{N}\}^{\mid W\mid \times T }$,
such that for $\omega \in \Omega$, $v \in W$ and $1 \leq t \leq T$, $\Pi_W(\omega) (v,t)$
indicates the state of node $v$ at time $t$ at the sample point $\omega$.
\end{definition}

\begin{lemma}[Independence]\label{claim:history}
Suppose $v\in V$ and  $\eta\in[0,1]$,
and let $W$ be any subset of $V$ with no directed path to $v$.
Then, we have $Pr[\theta_v\leq\eta\mid \Pi_{W}]=\eta$, under \NLT~model. %Model~\ref{model:avg}.
\end{lemma}

\begin{proof}
In \NLT~model, %Model~\ref{model:avg}, 
the randomness comes from the choices of thresholds $\theta=\{\theta_v\}_{v\in V}$.
The sample space is actually the set of all possible configurations of thresholds $\theta$.

Note that, the event $\theta_v\leq\eta$ is totally determined by the choice of $\theta_v$, 
and the value of $\Pi_W$ is totally determined by the choice of all $\theta_u$'s such that $u\neq v$. 
From the description of \NLT~model, the choices of thresholds are independent over different nodes.
This implies that, for any $Q$ in the range of $\Pi_W$, the events $\theta_v\leq\eta$ and $\Pi_W=Q$ are independent.

Hence, we get $Pr[\theta_v\leq\eta\mid\Pi_{W}]=Pr[\theta_v\leq\eta]=\eta$.
\qed \end{proof}

\subsection{Connection to The Random Walk}

Consider time $ t > 0$. The status of nodes at time $t$ may be associated, since they can share
some common ancestors and the threshold of such an ancestor affects the status of its descendants. 
This association between different nodes cause more complicates for the analysis. 
In order to assist the analysis and handle the association carefully, 
we introduce a random walk process and show that 
this random walk process share an interesting connection with \NLT~model. %Model~\ref{model:avg}. 
%This connection is revealed in Lemma~\ref{lemma:additive}.
Next, we introduce the random walk process. 

\begin{model}[Random Walk Process (\RW)]\label{model:R}
Consider an information network $G=(V,E,b)$. For any given
node $v\in V$, we define a random walk process as follows.
\begin{itemize}
\item
At time $t=0$, the walk starts at node $v$.
\item
Suppose at some time $t$, the current node is $u$.
A node $w\in \Gamma(u)$ is chosen with probability $b_{uw}$.
The walk moves to node $w$ at time $t+1$.\footnote{Observe that if $w$ is the void node,
then the walk remains at $w$.}
\end{itemize}
\end{model}

\begin{definition}[Reaching Event]\label{def:R}
%Consider an information network $G$ on which the Random Walk Process is run.
For any node $v$, subset $C \subseteq V$ and $t \geq 0$,
we use $R^t_v(C)$ to denote the event that a random walk starting from $v$ would
reach a node in $C$ at precisely time $t$.
\end{definition}

Next, we show the connection between the \NLT~model and the \RW~model for the case
when the permanent initial set is empty. The more general case (arbitrary permanent initial set)
will be considered later.

%Model~\ref{model:avg} and Model~\ref{model:R}.

%(Theorem~\ref{lemma:additive}) stronger than this claim.
\begin{lemma}[Connection between \NLT~model and \RW~model]\label{lemma:additive}
Consider an acyclic information network $G$, and let $v$ be a non-void node,
and $1 \leq t \leq T$.  On the same network $G$, consider the \NLT~process on a
transient initial set $A$ and the \RW~process starting at $v$.
Then, $\mathbb{E}[X^t_v(A)] = Pr[R^t_v(A)]$.
%
%
%\sum_{a \in A} \mathbb{E}[X^t_v(\{a\})]$, where $\mathbb{E}[X^t_v(\{a\})]$
%is the expected influence if the only initially active node $a$ is
%transient.
\end{lemma}
\begin{proof}
This lemma is the key point in our argument, for which the proof is not obvious. 
To assist the proof of Lemma~\ref{lemma:additive}, we next introduce a process
called ``Path Effect'' which is an ``equivalent'' presentation of the \NLT~model, 
and devote all the remaining part of this subsection to this lemma. 
\qed \end{proof}

We next introduce the \PE process which augments the \NLT~model and defines (random) 
auxiliary array structures $P^t_v$ known as the \emph{influence paths} to record the
influence history.
%, which is determined by the transient initial set
%$A$ and the configuration $\omega\in\Omega$ of all random parts.
%(Later it is shown that the present sample space $\Omega$ contains
%more randomness besides all possible threshold configurations.)
Intuitively, if $v$ becomes active at time $t$, then the path
$P^t_v$ shows which of the initially active nodes is
responsible.
%When there is no risk of confusion, we simply write
%$P^t_v$ for $P^t_v(A,\omega)$.
An important invariant is that node
$v$ is active at time $t$ if and only if $P^t_v[0] \in A$.

\begin{model}[Path Effect Process (\PE)]\label{model:P}
Consider an information network $G=(V,E,b)$.
Each node $v\in V$ is associated with a
threshold $\theta_v$, which is chosen from $(0,1)$ uniformly at random.

Given a transient initial set $A \subseteq V$ and a configuration of the thresholds
$\theta=\{\theta_v\}_{v\in V}$, for each node $v$ and each
time step, the influence paths $P^t_v$
are constructed in the following influence procedure.

\begin{itemize}
\item At time $t=0$, for any $v\in V$, $P^0_v[0]=v$.
\item At time $t>0$, define the active set at the previous time step
as $A_{t-1}=\{u\in V\mid P^{t-1}_u[0]\in A\}$.
For each node $v\in V$, we compute
$f_v(A_{t-1}) := \sum_{u\in \Gamma(v)\cap A_{t-1}}b_{vu}$. Then,
\begin{enumerate}
\item[a.]
If $f_v(A_{t-1}) \geq \theta_v$, choose node $u\in \Gamma(v)\cap A_{t-1}$
with probability
$\frac{b_{vu}}{f_v(A_{t-1})}$;
\item[b.]
If $f_v(A_{t-1})<\theta_v$, choose node $u \in \Gamma(v)\setminus A_{t-1}$
with probability
$\frac{b_{vu}}{1-f_v(A_{t-1})}$.
\end{enumerate}
Once $u$ is chosen, let $P^t_v[0,\ldots, t-1] := P^{t-1}_u[0,\ldots,t-1]$ and
$P^t_v[t]=u$.\footnote{Observe that if $v$ is the void node, then $P^t_v[0,\ldots,t]=[void,\ldots,void]$.}
\end{itemize}
\end{model}

\begin{remark}\label{rem:XI}
Observe that given an information network with a transient initial set $A$
and a configuration $\theta$ of thresholds. Both of the \NLT~model 
and the \PE~process produce exactly
the same active set $A_t$ at each time step $t$.
\end{remark}

\begin{definition}[Source Event] \label{def:I}
Consider an information network $G$ on which the \PE~process is run on the
initial active set $A$.
For any subset $C \subseteq V$, we use $I^t_v(C)$ to denote the event that
$P^t_v[0]$ belongs to $C$. If $C \subseteq A$, this event
means $v$'s state at time $t$ is the same as those of the nodes in $A$ at time 0 and hence is active.
We shall see later on that the event $I^t_v(C)$ is independent of $A$ and hence the notation has no dependence on $A$.
\end{definition}

When the given network is acyclic, the \PE~process has an interesting
property.
\begin{lemma}[Acyclicity Implies Independence of Choice]\label{lem:StW}
Consider an information network $G=(V,E,b)$ on which
the \PE~process is run with some initial active set.
If $G$ is acyclic, for any non-void node $v\in V$ and node $u\in\Gamma(v)$, we have
$Pr[P^t_v[t]=u\mid\Pi_{W}] = b_{vu}$,
where $W=\Gamma(v)$.  Recall that $\Pi_W$ carries the information about the states of the nodes in $W$
at every time step.
\end{lemma}

\begin{proof}
It is sufficient to prove that, for any value $Q$ in the range of $\Pi_{W}$,
$Pr[P^t_v[t]=u\mid \Pi_{W}=Q]=b_{vu}$
holds.
Because $f_v(A_{t-1})$ is determined by the states of $v$'s outgoing neighbors,
once $Q$ is fixed, $\eta(Q)=f_v(A_{t-1})$ is determined. We consider two cases.

\noindent \textbf{(1) When $u$ is in $A_{t-1}$ according to $Q$}.
We have

$Pr[P^t_v[t]=u\mid (\theta_v\leq\eta(Q))\cap(\Pi_{W}=Q)]
=\frac{b_{vu}}{\eta(Q)}$.

Since $u \in A_{t-1}$ holds, the event $\{P^t_v[t] = u\}$ implies that $\{ \theta_v \leq \eta(Q) \}$.
Hence, we have
\begin{eqnarray*}
Pr[P^t_v[t]=u\mid \Pi_W=Q] & =  & Pr[(P^t_v[t]=u) \cap (\theta_v \leq \eta(Q))\mid \Pi_W=Q] \\
& = & Pr[P^t_v[t]=u\mid (\theta_v\leq\eta(Q))\cap(\Pi_{W}=Q)] \cdot  \\
& & Pr[\theta_v\leq\eta(Q)\mid \Pi_W=Q].
\end{eqnarray*}

Because $G$ is acyclic, $W$ has no directed path to $v$.
By Lemma~\ref{claim:history}, we have
\begin{eqnarray*}
Pr[\theta_v\leq\eta(Q)\mid \Pi_{W}=Q]&=&\eta(Q).
\end{eqnarray*}
Hence, we have proved that $Pr[P^t_v[t]=u\mid \Pi_{W}=Q]=b_{vu}$.

\noindent \textbf{(2) When $u$ is not in $A_{t-1}$ according to $Q$}.
The proof of this case is similar to the previous one.
\qed \end{proof}

%From the above lemma, we directly get $Pr[P^t_v[t]=u]=b_{vu}$. This means in the Path Effect Process,
%at any time $t>0$, $v$ chooses his outgoing neighbor $u$ with probability
%$b_{vu}$.

Recall the events $I^t_v(C)$ and $R^t_v(C)$ introduced in
Definitions~\ref{def:I} and~\ref{def:R}.
%We next prove an important relation between Path Effect Process
%and Random Walk Process, which holds when
%the given information network is acyclic.
The following lemma
immediately implies Lemma~\ref{lemma:additive} with $C = A$,
and using the observation $\mathbb{E}[X^t_v(A)]=Pr[X^t_v(A)=1]=Pr[I^t_v(A)]$ from Remark~\ref{rem:XI}.

\begin{lemma}[Connection between the \PE~process and the \RW~process]\label{thm:IR}
Suppose the information network $G$ is acyclic, and $A$ is the transient initial set. For any
$C\subseteq V$, any non-void node $v\in V$ and $t\geq0$,
we have $Pr[I^t_v(C)]=Pr[R^t_v(C)]$.
In particular, the probability $Pr[I^t_v(C)]$ is independent of $A$.
\end{lemma}

\begin{proof}
We use induction on $t$.
For $t=0$, we have $Pr[I^0_v(C)]=1$ \emph{iff} $v\in C$ and
$Pr[R^0_v(C)]=1$ \emph{iff} $v\in C$. Hence, $Pr[I^0_v(C)]=Pr[R^0_v(C)]$.

Suppose $Pr[I^t_v(C)]=Pr[R^t_v(C)]$ holds for
all $v\in V$ at any time $t<k$.

We consider the case $t=k$ and fix some non-void node $v$. Let $W := \Gamma(v)$.
Recall that the random object $\Pi_W$ carries information about the states
of all $v$'s outgoing neighbors at every time step. Let
$\mathbb{C}^{k-1}_u$ be the set of the values for $\Pi_W$ under which
$P^{k-1}_u[0]\in C$.

Observing that the events
$\Pi_W=Q$ for different $Q$'s are mutually exclusive, we have
\[Pr[I^k_v(C)]=\sum_{u\in\Gamma(v)}\sum_{Q\in\mathbb{C}^{k-1}_u}
Pr[\Pi_{W}=Q\mid P^k_v[k]=u] \cdot Pr[P^k_v[k]=u].\]

By Lemma~\ref{lem:StW}, we have
$Pr[P^t_v[t]=u\mid \Pi_{W}]=b_{vu}=Pr[P^t_v[t]=u]$, which implies for any value $Q$ of $\Pi_{W}$,
it holds that $Pr[\Pi_{W}=Q\mid P^k_v[k]=u]=Pr[\Pi_{W}=Q]$. Consequently, we have
\begin{eqnarray*}
Pr[I^k_v(C)]
&=&\sum_{u\in\Gamma(v)}\sum_{Q\in\mathbb{C}^{k-1}_u}
Pr[\Pi_{W}=Q] \cdot Pr[P^k_v[k]=u]\\
&=&\sum_{u\in\Gamma(v)}Pr[I^{k-1}_u(C)] \cdot Pr[P^k_v[k]=u].
\end{eqnarray*}

By induction hypothesis $Pr[I^{k-1}_u(C)]=Pr[R^{k-1}_u(C)]$, we get,
\begin{eqnarray*}
Pr[I^k_v(C)]&=&\sum_{u\in\Gamma(v)}Pr[R^{k-1}_u(C)] \cdot Pr[P^k_v[k]=u]\\
&=&\sum_{u\in\Gamma(v)}Pr[R^{k-1}_u(C)]b_{vu}.
\end{eqnarray*}

The last term is just $R^k_u(C)$, according to the description of the \RW~process. 
This completes the inductive step of the proof.
\qed \end{proof}

%We are now ready to prove
%Lemma~\ref{lemma:additive}.

%\begin{theorem}
%Consider the non-progressive Linear Threshold Model. If the given information network
%$G$ is acyclic, $\mathbb{E}[X^t_v(A)]$ is additive as a function of $A$.
%\end{theorem}

%
%
%\begin{proofof}{Lemma~\ref{lemma:additive}}
%Based on Remark~\ref{rem:XI},
%we have $\mathbb{E}[X^t_v(A)]=Pr[X^t_v(A)=1]=Pr[I^t(A,v)]$. Then,
%it suffices to show that,
%$Pr[I^t(A,v)]=\sum_{u\in A}Pr[I^t(\{u\},v)]$. This is easy to prove, by
%observing that $Pr[I^t(A,v)]=Pr[R^t_v(A)]$ (Lemma~\ref{thm:IR})
%and $Pr[R^t_v(A)]=\sum_{u\in A}Pr[R^t_v(\{u\})]$.
%\end{proofof}

%
%
%Recalling Remark~\ref{rem:compute_R}, we can compute $\mathbb{E}[X^t_v(A)]$ by considering the
%transition matrix $M$ of the Random Walk Process.

In the next subsection, we will reduce the general case with non-empty permanent initial set 
to the case when only transient initial set. Furthermore, 
we can prove the final conclusion (Theorem~\ref{thm:2-approx-Avg}).

\subsection{Submodularity of Acyclic \NLT}

At first, we consider the case where the permanent initial
set $\widehat{A}$ is non-empty. We show that this general case can be reduced
to the case where only transient initial set $A$ is non-empty, by the following transformation.
Suppose $G$ is an information network, with transient
initial set $A$ and permanent initial set
$\widehat{A}$, and
$T$ is the number of time steps to be considered. Consider
the following transformation on the network instance.  For
each node $y \in \widehat{A}$, do the following:
\begin{enumerate}
\item Add a chain $D_y$
of $T$ dummy nodes to the network: starting from the head node of the chain, exactly
one edge with weight 1 points to the next node, and so on, until the end node is reached.
\item Remove all outgoing edges from $y$. Add exactly one outgoing edge with weight 1 from $y$
to the head of the chain $D_y$
\end{enumerate}
See Fig~\ref{fig2} for an example of the chain of dummy nodes.
\begin{figure}
\centering
\includegraphics[width=0.5\textwidth]{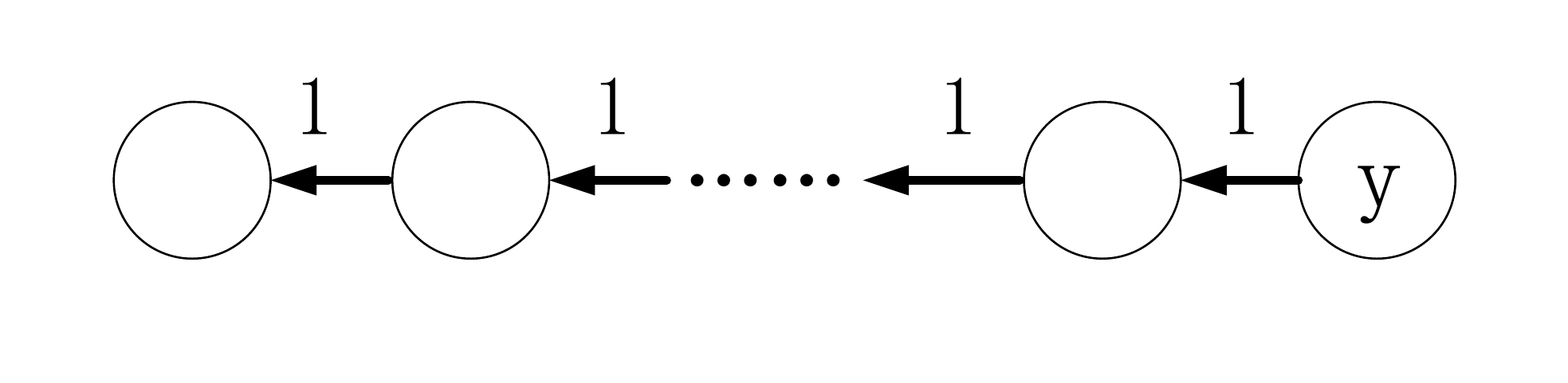}
\caption{The chain of dummy nodes}
\label{fig2}
\end{figure}
Let $D := \cup_{y \in \widehat{A}} D_y$ be the set of dummy nodes. 
We call the new network $\overline{G}(\widehat{A})$ the transformed network of $G$ with
respect to $\widehat{A}$.
When there is no risk of confusion, we simply write $\overline{G}$. 
The transformed instance on $\overline{G}(\widehat{A})$ only has
 $A \cup \widehat{A} \cup D$ as the transient initial set
and no permanent initial node.
%Instead of running Model~\ref{model:I} on $G$ with transient initial set $A$
%and permanent initial set $\widehat{A}$,
%we run the model on $\overline{G}[\widehat{A}]$ with only
%transient initial set $A \cup \widehat{A} \cup D$.
The initially active dummy nodes in $D$ ensure that every node $y \in \widehat{A}$ is active for $T$ time steps.
We use the notation convention that we add an overline
to a variable (e.g., $\overline{X}$), if it is associated with the transformed network.
\begin{remark} \label{remark:transform_equiv}
For any non-dummy, non-void node $v$,
\[X^t_v(A, \widehat{A}) = \overline{X}^t_v(A \cup \widehat{A} \cup D).\]
\end{remark}
\begin{lemma}\label{lemma:reaching} Suppose we are given an instance on information network
 $G$, with transient initial set $A$ and permanent initial set $\widehat{A}$.  
Let $v$ be any non-void node in $G$ and $0 \leq t \leq T$.  Suppose in the
transformed network $\overline{G}(\widehat{A})$, for
any subset $C$ of nodes in $G$, $\overline{R}^t_v( C)$ is the event that starting
at $v$, the \RW process on $\overline{G}$ for $t$ steps ends at
a node in $C$. Then,
\begin{eqnarray*}
\mathbb{E}[X^t_v(A,\widehat{A})] = \sum_{u\in A}Pr[\overline{R}^t_v(\{u\})]
+\sum_{y\in\widehat{A}}\sum^t_{i=0}Pr[\overline{R}^i_v(\{y\})].
\end{eqnarray*}
\end{lemma}
\begin{proof}
Let $v$ be a non-void node in $G$ and hence cannot be a dummy node
in $\overline{G}$. By lemma~\ref{lemma:additive}, 
the equation $\mathbb{E}[X^t_v(A, \widehat{A})] = \mathbb{E}[\overline{X}^t_v(A \cup \widehat{A} \cup D)]$ implies
\[\mathbb{E}[X^t_v(A, \widehat{A})] = \sum_{u\in A}Pr[\overline{R}^t_v(\{u\})]
+\sum_{y \in \widehat{A}} \left(Pr[\overline{R}^t_v(\{y\})] +
\sum_{w \in D_y} Pr[\overline{R}^t_v(\{w\})]\right).\]

Consider the \RW~process on $\overline{G}$ starting at $v$. 
For any node $y\in\widehat{A}$, and consider a node $w \in D_y$
that is $i$ hops away from $y$.  If $i > t$, then it is impossible
for $v$ to reach $w$ in $t$ steps.  Observe that if $v$ reaches $w$
at time $t$, then $v$ must reach $y$ at time $t-i$.  Hence,
$Pr[\overline{R}^t_v(\{w\})]=Pr[\overline{R}^{t-i}_v(\{y\})]$, and the summation
over $i$ from 1 to $t$ gives the required formula.
\qed \end{proof}
%in the chain attached to
%it, consider a node $w$ such that the path from $y$ to $w$ is of length $i$.
%Note that if $i \geq t$, any Random Walk starting from $v$ cannot touch $w$
%within $t$ steps. (See Fig~\ref{fig4}.)
%
%\begin{figure}
%    \centering
%    \includegraphics[width=0.6\textwidth]{Fig4.eps}
%    \caption{a path of length $i$}
%    \label{fig4}
%\end{figure}

%
%By the construction of the dummy nodes chain, we have
%$Pr[R^t(v,\{w\})]=Pr[R^{t-i}(v,\{y\})]$. Therefore,
%\begin{eqnarray*}
%Pr[X^t_v(A,\hat{A})]&=&\sum_{u\in A}Pr[R^t_v(\{u\})]
%+\sum_{y\in\widehat{A}}\sum^{t-1}_{i=0}Pr[R^{t-i}(v,\{y\})]\\
%&=&\sum_{u\in A}Pr[R^t_v(\{u\})]
%+\sum_{y\in\widehat{A}}\sum^{t}_{i=1}Pr[R^{i}(v,\{y\})].
%\end{eqnarray*}

\begin{definition}[Passing-Through Event] \label{defn:passing}
Let $G$ be an information network. For any node $v$, subset $C \subseteq V$ and $t \geq 0$,
we use $S^t_v(C)$ to denote the event that a \RW~process on $G$ starting from $v$ would
reach a node in $C$ at time $t$ \textbf{or before}.
\end{definition}

\begin{lemma}[General Connection between the \NLT~model and the \RW~model] \label{lemma:general_con}
Suppose $G$ is an acyclic information network,
and let $v$ be a non-void node, and $1 \leq t \leq T$.
On the same network $G$, consider the \NLT~model with 
transient initial set $A$ and permanent initial set $\widehat{A}$,
and the \RW~process starting at $v$.
Then, $\mathbb{E}[X^t_v(A, \widehat{A})] = Pr [R^t_v(A) \cup S^t_v(\widehat{A})]$.
%Observe this holds even when $A \cap \widehat{A} \neq \emptyset$.
%
%Consider the Random Walk Process on $G$ for $t$ steps starting
%at $v$.  For any subset $S$ of nodes,
%let $E(S)$
%
% consider
%a random walk of length $t$ \textbf{in $G$} starting at $v$. For any
%transient initial set $A$ and permanent initial set $\widehat{A}$,
%$\mathbb{E}[X^t_v(A, \widehat{A})]$ is the probability
%that the walk ends at $A$ or passes through $\widehat{A}$.
%
%In another words, for a random walk of length $t$ \textbf{in $G$} starting from $v$,
%define two events as,
%\begin{eqnarray*}
%E^1(S):&&\textrm{ The walk ends at $S$},\\
%E^2(S):&&\textrm{ The walk passes through $S$}.
%\end{eqnarray*}
%
%Then for any $A$ and $\widehat{A}$,
%$\mathbb{E}[X^t_v(A, \widehat{A})]=Pr[E^1(A)\cup E^2(\widehat{A})]$.
\end{lemma}
\begin{proof}\label{ref:XE}
Without loss of generality,
we can still assume $A \cap \widehat{A} = \emptyset$,
because $X^t_v(A, \widehat{A}) = X^t_v(A \setminus \widehat{A}, \widehat{A})$ and $R^t_v(A) \cup S^t_v(\widehat{A}) = R^t_v(A \setminus \widehat{A}) \cup S^t_v(\widehat{A})$.
  From Lemma~\ref{lemma:reaching},
we have
\begin{eqnarray*} \label{eq:one}
\mathbb{E}[X^t_v(A,\widehat{A})] = \sum_{u\in A}Pr[\overline{R}^t_v(\{u\})]
+\sum_{y\in\widehat{A}}\sum^t_{i=0}Pr[\overline{R}^i_v(\{y\})],
\end{eqnarray*}

where the notation $\overline{R}$ means the corresponding
term referring to the reaching event in the
transformed graph $\overline{G} := \overline{G}[\widehat{A}]$.

We compare the random walks of $t$ steps starting at $v$ on $G$ and on $\overline{G}$ using
a coupling argument.  Starting at $v$, the random walk
on $\overline{G}$ copies the random choices made in $G$.  This
goes smoothly for the walk on $\overline{G}$ until
a node $y$ in $\widehat{A}$ is hit, at which point further
random choices made in $G$ are irrelevant.
From this coupling argument, we can relate the events from $G$ and $\overline{G}$ in the following way:
\begin{itemize}
\item For $u \in A$, $Pr[\overline{R}^t_v( \{u\})] = Pr[R^t_v( \{u\}) \setminus S^t_v( \widehat{A})]$.
\item For $y \in \widehat{A}$, $\sum^t_{i=0}Pr[\overline{R}^i_v(\{y\})]$ is the probability
that the walk in $G$ hits $y$ before any other node in $\widehat{A}$.
\end{itemize}
Hence, it follows that on the right hand side of (\ref{eq:one}),
the first term is $Pr[Pr[R^t_v( A) \setminus S^t_v( \widehat{A})]$
and the second term is $Pr[S^t_v( \widehat{A})]$. Therefore,
their sum is  $Pr [R^t_v(A) \cup S^t_v(\widehat{A})]$, as required.
%
%
%
%Observe that for $u\in A$
%the event $R^t_v(\{u\})$ of the random walk in $\overline{G}[\widehat{A}]$ is
%actually the event ``end at $u$ but do not pass through $\widehat{A}$''
%of the random walk in $G$. Then the event $\cup_{u\in A}R^t_v(\{u\})$ is
%the event ``end at $A$ but doe not pass through $\widehat{A}$'' in $G$.
%That is,
%\begin{eqnarray*}
%\sum_{u\in A}Pr[R^t_v(\{u\})&=&Pr[\cup_{u\in A}R^t_v(\{u\})]\\
%&=&Pr[E^1(A)\setminus E^2(\widehat{A})].
%\end{eqnarray*}
%
%For $y\in\widehat{A}$,
%the event $R^i(v,\{y\})$ of the random walk in $\overline{G}[\widehat{A}]$ is
%actually the event ``reach $u$ at time $i$ but
%do not pass through $\widehat{A}$ before time $i$''
%of the random walk in $G$. Then the event $\cup_{y\in\widehat{A}}R^i(v,\{u\})$ is
%the event ``the first touch to $\widehat{A}$ happens at time $i$'' in $G$.
%Therefore, we have,
%\begin{eqnarray*}
%\sum_{y\in\widehat{A}}\sum^{t}_{i=1}Pr[R^{i}(v,\{y\})]&=&
%\sum^{t}_{i=1}Pr[\cup_{y\in\widehat{A}}R^{i}(v,\{y\})]\\
%&=&Pr[E^2(\widehat{A})].
%\end{eqnarray*}
%
%Thus it is true that,
%\begin{eqnarray*}
%\mathbb{E}[X^t_v(A, \widehat{A})]&=&Pr[E^1(A)\cup E^2(\widehat{A})]+Pr[E^2(\widehat{A})]\\
%&=&Pr[E^1(A)\cup E^2(\widehat{A})].
%\end{eqnarray*}
\qed \end{proof}

\ignore{
\begin{remark}\label{rem:compute_general}
\textbf{(Computing $\mathbb{E}[X^t_v(A,\widehat{A})]$).}
Given an information network, we can use the modified transition matrix $\overline{M}$
such that each node in $\widehat{A}$ corresponds to
an absorbing state.
Then, $x^t_v(A,\widehat{A}) =
\sum_{u \in A \cup \widehat{A}} \overline{M}^t(v,u)$, 
and $\mathbb{E}[X^t_v(A,\widehat{A})] = Pr[x^t(A,\widehat{A})\geq \theta_v] = F(x^t_v(A,\widehat{A})),$
where $F(x)$ is the cumulative distribution function from which $\theta$'s are selected.
\end{remark}
}
\begin{theorem}\label{thm:sub}
{\bf (Submodularity and Monotonicity of $\mathbb{E}[X^t_v(A,\widehat{A})]$).}
Consider the \NLT~model on an acyclic information network $G$ with transient initial
set $A$ and permanent initial set $\widehat{A}$.
Then, the function $(A, \widehat{A}) \mapsto \mathbb{E}[X^t_v(A,\widehat{A})]$ is submodular and monotone.
\end{theorem}
\begin{proof}
For notational convenience,
we drop the superscript $t$ and the subscript $v$,
and write for instance $X(A, \widehat{A}) := X^t_v(A, \widehat{A})$.
For the reaching and the passing-through events associated
with the Random Walk Process in $G$, we write
$R(A) := R^t_v(A)$ and $S(A) := S^t_v(A)$

It is sufficient to prove that, for any $A\subseteq B\subseteq V$,
$\widehat{A}\subseteq \widehat{B}\subseteq V$, and node $w\not\in (B\cup\widehat{B})$,
the following inequalities hold:
\begin{eqnarray}
\mathbb{E}[X(A\cup\{w\},\widehat{A})] - \mathbb{E}[X(A,\widehat{A})]
\geq \mathbb{E}[X(B\cup\{w\},\widehat{B})] - \mathbb{E}[X(B,\widehat{B})];\label{eqn:sub1}\\
\mathbb{E}[X(A,\widehat{A}\cup\{w\})] - \mathbb{E}[X(A,\widehat{A})]
\geq \mathbb{E}[X(B,\widehat{B}\cup\{w\})] - \mathbb{E}[X(B,\widehat{B})].\label{eqn:sub2}
\end{eqnarray}
By Lemma~\ref{lemma:general_con}, for any subsets $C$ and  $\widehat{C}$ such that $w\not\in (C\cup\widehat{C})$,
$x(C\cup\{w\},\widehat{C})-x(C,\widehat{C}) =
Pr[ R(C\cup\{w\}) \cup S(\widehat{C})] - Pr[ R(C) \cup S(\widehat{C})] = Pr[R(\{w\}) \setminus S(\widehat{C})]$,
where the last equality follows from definitions of reaching
and passing-through events.  Hence,
inequality (\ref{eqn:sub1}) follows because
$\widehat{A} \subseteq \widehat{B}$ implies that
$R(\{w\}) \setminus S(\widehat{A}) \supseteq R(\{w\}) \setminus S(\widehat{B})$.

Similarly, $\mathbb{E}[X(C,\widehat{C} \cup\{w\})] - \mathbb{E}[X(C,\widehat{C})] = Pr[S(\{w\}) \setminus (R(C) \cup S(\widehat{C}))]$.  Hence, inequality (\ref{eqn:sub2}) follows
because $R(A) \cup S(\widehat{A}) \subseteq R(B) \cup S(\widehat{B})$.
%
%
%
%
%
%
%
%
%
%
%Based on Lemma~\ref{ref:XE},
%\begin{eqnarray*}
%&&\mathbb{E}[X^t_v(A\cup\{x\},\widehat{A})]-\mathbb{E}[X^t_v(A,\widehat{A})]\\
%&=&Pr[E^1(A\cup\{x\})\cup E^2(\widehat{A})]-Pr[E^1(A)\cup E^2(\widehat{A})]\\
%&=&Pr[E^1(\{x\})\setminus E^2(\widehat{A})],\\
%&&\mathbb{E}[X^t_v(B\cup\{x\},\widehat{B})]-\mathbb{E}[X^t_v(B,\widehat{B})]\\
%&=&Pr[E^1(B\cup\{x\})\cup E^2(\widehat{B})]-Pr[E^1(B)\cup E^2(\widehat{B})]\\
%&=&Pr[E^1(\{x\})\setminus E^2(\widehat{B})].
%\end{eqnarray*}
%$\widehat{A}\subseteq \widehat{B}$ implies $E^2(\widehat{A})\subseteq E^2(\widehat{B})$,
%thus we know (~\ref{eqn:sub1}) is correct.
%
%Then consider the remaining part of the conclusion.
%\begin{eqnarray*}
%&&\mathbb{E}[X^t_v(A,\widehat{A}\cup\{x\})]-\mathbb{E}[X^t_v(A,\widehat{A})]\\
%&=&Pr[E^1(A)\cup E^2(\widehat{A}\cup\{x\})]-Pr[E^1(A)\cup E^2(\widehat{A})]\\
%&=&Pr[E^2(\{x\})\setminus (E^1(A)\cup E^2(\widehat{A}))],\\
%&&\mathbb{E}[X^t_v(B,\widehat{B}\cup\{x\})]-\mathbb{E}[X^t_v(B,\widehat{B})]\\
%&=&Pr[E^1(B)\cup E^2(\widehat{B}\cup\{x\})]-Pr[E^1(B)\cup E^2(\widehat{B})]\\
%&=&Pr[E^2(\{x\})\setminus (E^1(B)\cup E^2(\widehat{B}))].
%\end{eqnarray*}
%$A\subseteq B$, $\widehat{A}\subseteq \widehat{B}$
%imply $E^1(A)\cup E^2(\widehat{A})\subseteq E^1(B)\cup E^2(\widehat{B})$,
%thus we know (~\ref{eqn:sub2}) is correct.
\qed \end{proof}

\begin{corollary}[Objective Function is Submodular and Monotone]
With the same hypothesis as in Theorem~\ref{thm:sub},
the function $(A, \widehat{A}) \mapsto \overline{\sigma}(A,\widehat{A})$ is submodular and monotone.
\end{corollary}
At the end, we achieve the main result of this paper.

\begin{theorem}\label{thm:2-approx-Avg}
%Consider the Non-progressive Threshold Model (Model~\ref{model:avg} introduced in Section~\ref{sec:pre}).
Given an acyclic information network, a time period $[1,T]$, a budget $K$ and advertising costs (transient or permanent) that are uniform over the nodes,
an advertiser can use the Standard Greedy Algorithm to compute a transient initial set $A$ and 
a permanent initial set $\widehat{A}$ with total cost at most $K$ in polynomial time
such that $\overline{\sigma}(A, \widehat{A})$ is at least $\frac{1}{2}$ of the optimal value.  
Moreover, there is a randomized algorithm that outputs $A$ and $\widehat{A}$ such that
the expected value (over the randomness of the randomized algorithm)
of $\overline{\sigma}(A, \widehat{A})$ is at least $1 - \frac{1}{e}$ of the optimal value, where $e$ is the natural number.
\end{theorem}
\begin{proof}
We describe how Theorem~\ref{thm:2-approx-Avg} is derived.
Recall that the advertiser is given a budget $K$, and the cost per transient node is $c$ and the cost per permanent node is $\widehat{c}$.
Observe that if the advertiser uses $k$ transient nodes, where
$k \leq \floor{\frac{K}{c}}$, then there can be at most
$\widehat{k} := \floor{\frac{K - k c}{\widehat{c}}}$
permanent nodes.  Hence, for each such guess of $k$ and the
corresponding $\widehat{k}$, the advertiser just needs to consider the maximization of the
submodular and monotone function $(A, \widehat{A}) \mapsto \overline{\sigma}(A, \widehat{A})$
on the matroid $\{(A, \widehat{A}) : |A| \leq k, |\widehat{A}| \leq
\widehat{k}\}$, for which $\frac{1}{2}$-approximation can be obtained
in polynomial time using the techniques of Fisher et al.\ \cite{Fisher1978}. 
A randomized algorithm given by Calinescu et al.\ \cite{DBLP:journals/siamcomp/CalinescuCPV11} achieves
$(1-\frac{1}{e})$-approximation in expectation.
\qed \end{proof}